\documentclass{article}%
\usepackage{amssymb}
\usepackage{amsfonts}
\usepackage{amsmath}
\usepackage{graphicx}
\usepackage{amssymb,color}%
\providecommand{\U}[1]{\protect\rule{.1in}{.1in}}

\definecolor{c20}{rgb}{0.,0.0,0.}
\definecolor{c30}{rgb}{0.,0.,0.}
\definecolor{c40}{rgb}{0,0.0,0.0}
\definecolor{c50}{rgb}{0,0,0}

\newtheorem{theorem}{Theorem}

\newtheorem{case}{Case}

\newtheorem{conclusion}{Conclusion}

\newtheorem{corollary}{Corollary}

\newtheorem{example}{Example}

\newenvironment{proof}[1][Proof]{\noindent\textbf{#1.} }{\ \rule{0.5em}{0.5em}}
\begin{document}

\title{A note on Fibonacci Sequences of Random Variables }
\author{Ismihan Bayramoglu \\
Department of Mathematics, Izmir University of Economics, Izmir, Turkey\\
E-mail: ismihan.bayramoglu@ieu.edu.tr}
\maketitle

\begin{abstract}
The focus of this paper is the random sequences in the form $\{X_{0},X_{1},$
$X_{n}=X_{n-2}+X_{n-1},n=2,3,..\dot{\}},$ referred to as Fibonacci Sequence
of Random \ Variables (FSRV). The initial random variables $X_{0}$ and $%
X_{1} $ are assumed to be absolutely continuous with joint probability
density function (pdf) $f_{X_{0},X_{1}}.$ The FSRV is completely determined
by $X_{0} $ and $X_{1}$ and the members of Fibonacci sequence $\digamma
\equiv \{0,1,1,2,3,5,8,13,21,34,55,89,144,...\}.$ \ We examine the
distributional and limit properties of the random sequence $%
X_{n},n=0,1,2,... $ . \

\end{abstract}

\section{Introduction}

Let $\{\Omega ,\digamma ,P\}$ be a probability space and $X_{i}\equiv
X_{i}(\omega ),\omega \in \Omega ,i=0,1$ be absolutely continuous random
variables defined on this probability space with joint probability density
function (pdf) $f_{X_{0},X_{1}}(x,y).$ Consider a sequence of random
variables $X_{n}\equiv X_{n}(\omega ),n\geq 1$ given in $\{\Omega ,\digamma
,P\}$ defined as $\{X_{0},X_{1},$ $X_{n}=X_{n-2}+X_{n-1},$ $n=2,3,..\}.$ We
call this sequence \ "the Fibonacci Sequence of Random Variables". \ It is
clear that $X_{2}=X_{0}+X_{1},$ $X_{3}=X_{0}+2X_{1},...$ and for any $%
n=0,1,2,...$ we have $X_{n}=a_{n-1}X_{0}+a_{n}X_{1},$ where \ $\
\{a_{n}=a_{n-2}+a_{n-1},n=2,3,...;a_{0}=0,a_{1}=1,a_{2}=1$ $\}$ is the
Fibonacci sequence $\digamma \equiv \{0,1,1,2,3,5,8,13,21,34,55,89,144,...\}.
$ It is also clear that the Fibonacci Sequence of Random Variables (FSRV) $%
X_{n},n=0,1,2,...$ is the sequence of dependent random variables based on
initial random variables $X_{0}$ and $X_{1},$ which fully defined by the
members of the Fibonacci sequence $\digamma $. \ We are interested in the
behavior of FSRV, i.e. the distributional properties of $X_{n}$ and joint
distributions of $X_{n}$ and $X_{n+k\text{ }}$ for any $n$ and $k.$ \ In the
Appendix Figure A1 and Figure A2, we present some examples of realizations
of FSRV in the case of independent random variables $X_{0}$ and $X_{1}$
having \textit{\ Uniform(0,1)} distribution and Standard normal distribution
with the R codes provided.

This paper is organized as follows. In Section 2, the probability density
function of $X_{n}$ is considered, followed by a discussion of two cases
where $X_{0}$ and $X_{1}$ have exponential and uniform distributions,
respectively. \ Then, there is an investigation of limit behavior of ratios
of some characteristics of \ pdf of $X_{n}$ for large $n.$ In the considered
examples, the ratio of maximums of the pdfs, modes and expected values of
consecutive elements of FSRV converge to golden ratio $\varphi \equiv \frac{%
1-\sqrt{5}}{2}=1,6180339887...$ . \ The ratio $X_{n+1}/X_{n}$ and normalized
sums \ of $X_{n}$'s for large $n$ are discussed in Section 3. In Section 4,
\ the focus is on the joint distributions of \ $X_{n}$ and $X_{n+k},$ for $%
2\leq k\leq n$ and on the prediction of $X_{n+k}$ given $X_{n}.$

\section{Distributions}

Consider $\ X_{n}=a_{n-1}X_{0}+a_{n}X_{1},$ $n=0,1,2,...$ , where $X_{0}$
and $X_{1}$ are absolutely continuous random variables with joint pdf $%
f_{X_{0},X_{1}}(x,y),$ $(x,y)\in
\mathbb{R}
^{2}$ and $a_{n},n=0,1,2,...$ is the Fibonacci sequence. \ Denote by $f_{0}$
and $f_{1}$ the marginal pdf's of $X_{0}$ and $X_{1},$ respectively.

\begin{theorem}
\label{Theorem 1}The pdf of $X_{n}$ is
\begin{equation}
f_{X_{n}}(x)=\frac{1}{a_{n}a_{n-1}}\int\limits_{-\infty }^{\infty
}f_{X_{0},X_{1}}(\frac{x-t}{a_{n-1}},\frac{t}{a_{n}})dt.  \label{a000}
\end{equation}
\end{theorem}

If $X_{0}$ and $X_{1}$ are independent, then
\begin{equation}
f_{X_{n}}(x)=\frac{1}{a_{n}a_{n-1}}\int\limits_{-\infty }^{\infty
}f_{X_{0}}(\frac{x-t}{a_{n-1}})f_{X_{1}}(\frac{t}{a_{n}})dt.  \label{a00}
\end{equation}

\begin{proof}
Equations (\ref{a000}) and (\ref{a00}) \ are straightforward results \ of
distributions of linear functions of random variables (see eg., Feller
(1971), Ross (2002), Gnedenko (1978), Skorokhod (2005)) \
\end{proof}

\begin{case}
\label{Case 1}Exponential distribution. Let \ $X_{0}$ and $X_{1}$ \ be
independent and identically distributed (iid) random variables having
exponential distribution with parameter $\lambda =1.$ Then the pdf of $X_{n}$
is
\begin{eqnarray}
f_{X_{n}}(x) &=&\frac{1}{a_{n-2}}\left\{ \exp \left( \frac{xa_{n-2}}{%
a_{n-1}a_{n}}\right) -1\right\} \exp (-\frac{x}{a_{n-1}}),x\geq 0,\text{ }%
n=3,4,...\text{ }  \label{a001} \\
f_{X_{2}}(x) &=&x\exp (-x),x\geq 0.  \notag
\end{eqnarray}%
Below in Figure 1, the graphs of $f_{X_{n}}(x)$ for different values of $n$\
are presented.

\includegraphics[scale=0.22]{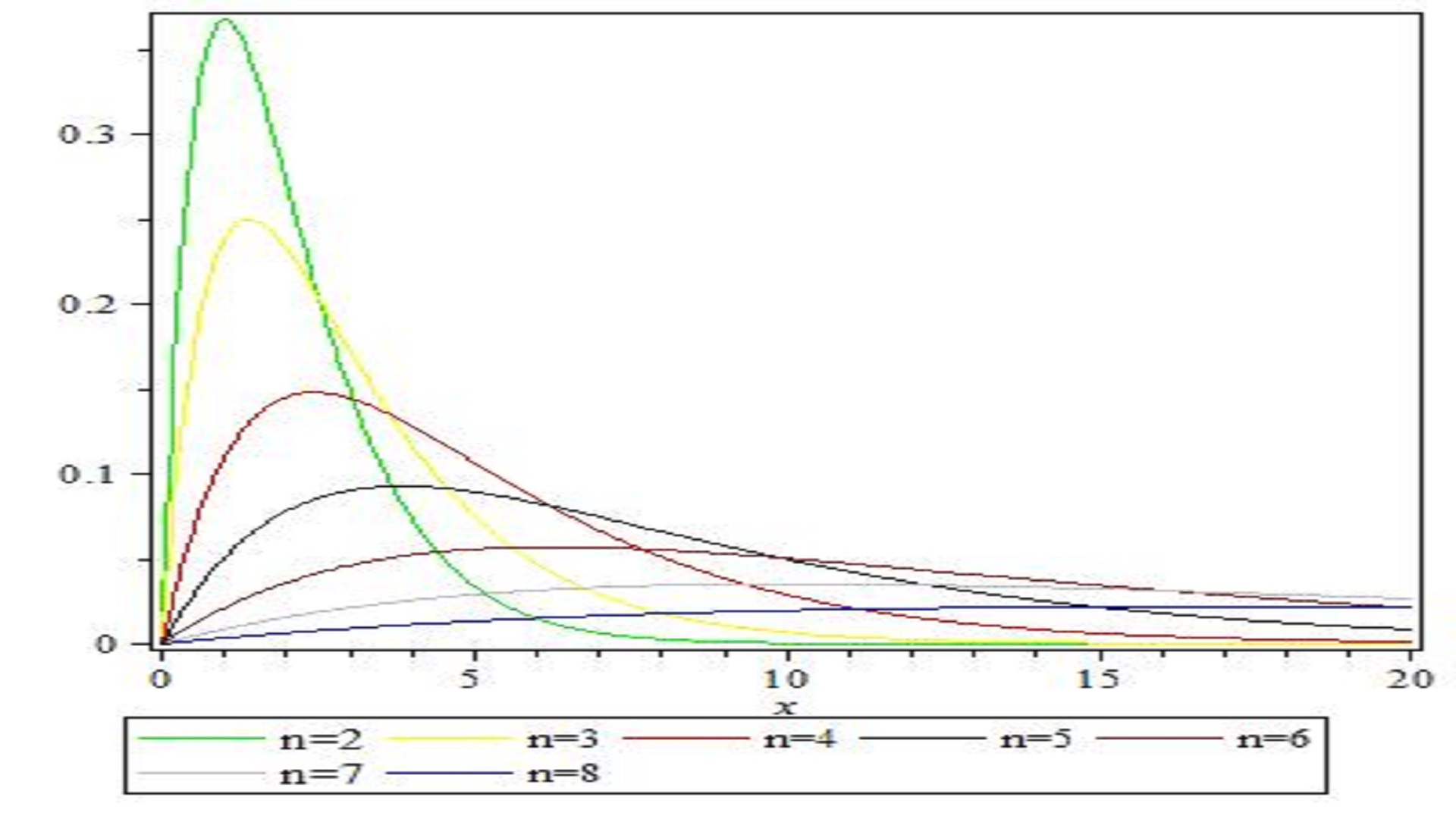}

\begin{tabular}{l}
Figure 1. Graphs of $f_{X_{n}}(x),$ $n=2,3,4,5,6,7,8,$ given in (\ref{a001})%
\end{tabular}%
The expected value of $X_{n}$ is%
\begin{eqnarray*}
EX_{n} &=&\frac{1}{a_{n-2}}\left( \int\limits_{0}^{\infty }x\exp \left(
-x\left( \frac{a_{n}-a_{n-2}}{a_{n-1}a_{n}}\right) \right)
dx-\int\limits_{0}^{\infty }x\exp (-\frac{x}{a_{n-1}})\right) dx \\
&=&\frac{1}{a_{n-1}}\left( \frac{a_{n}^{2}a_{n-1}^{2}}{(a_{n}-a_{n-2})^{2}}%
-a_{n-1}^{2}\right) =a_{n+1}.
\end{eqnarray*}%
and variance is
\begin{eqnarray*}
Var(X_{n}) &=&\frac{1}{a_{n-2}}\left( \int\limits_{0}^{\infty }x^{2}\exp
\left( -x\left( \frac{a_{n}-a_{n-2}}{a_{n-1}a_{n}}\right) \right)
dx-\int\limits_{0}^{\infty }x^{2}\exp (-\frac{x}{a_{n-1}})\right) dx \\
&=&a_{2n-1}.
\end{eqnarray*}
\end{case}

\begin{theorem}
\label{Theorem 2} Let $M_{n}=\underset{0<x<\infty }{\max }f_{X_{n}}(x)$ and $%
x_{n}^{\ast }=\underset{0<x<\infty }{\arg \max }f_{X_{n}}(x)$ be the maximum
of $f_{X_{n}}(x)$ and mode of $\ X_{n},$ $n=2,3,...,$ respectively. Then
\begin{eqnarray*}
\underset{n\rightarrow \infty }{\lim }\frac{M_{n}}{M_{n+1}} &=&\underset{%
n\rightarrow \infty }{\lim }\frac{x_{n+1}^{\ast }}{x_{n}^{\ast }}=\underset{%
n\rightarrow \infty }{\lim }\frac{E(X_{n+1})}{E(X_{n})}=\varphi \text{ and }
\\
\underset{n\rightarrow \infty }{\lim }\frac{Var(X_{n+1})}{Var(X_{n})}
&=&\varphi ^{2},
\end{eqnarray*}%
where%
\begin{equation*}
\varphi \equiv \frac{1-\sqrt{5}}{2}=1,6180339887...
\end{equation*}%
is the golden ratio.
\end{theorem}

\begin{proof}
The following can easily be verified
\begin{equation}
\frac{d}{dx}f_{X_{n}}(x)=\frac{(-\frac{x}{a_{n-1}})(e^{\frac{xa_{n-2}}{%
a_{n-1}a_{n}}}-1)}{a_{n-2}a_{n-1}}+\frac{e^{-\frac{x}{a_{n-1}}}e^{\frac{%
xa_{n-2}}{a_{n-1}a_{n}}}}{a_{n-1}a_{n}}=0.  \label{a010}
\end{equation}%
The equation (\ref{a010}) has unique solution
\begin{equation*}
x_{n}^{\ast }=\frac{a_{n-1}a_{n}\ln \left( \frac{a_{n}}{a_{n}-a_{n-2}}%
\right) }{a_{n-2}}.
\end{equation*}%
Therefore $X_{n}$ is unimodal and we have
\begin{eqnarray*}
M_{n} &=&f_{X_{n}}(x_{n}^{\ast })=\frac{1}{a_{n}-a_{n-2}}\left( \frac{a_{n}}{%
a_{n}-a_{n-2}}\right) ^{-\frac{a_{n}}{a_{n-2}}} \\
M_{n+1} &=&f_{X_{n+1}}(x^{\ast })=\frac{1}{a_{n+1}-a_{n-1}}\left( \frac{%
a_{n+1}}{a_{n+1}-a_{n-1}}\right) ^{-\frac{a_{n+1}}{a_{n-1}}}
\end{eqnarray*}%
and using%
\begin{equation*}
\underset{n\rightarrow \infty }{\lim }\frac{a_{n+\alpha }}{a_{n}}=\varphi
^{\alpha }
\end{equation*}%
we obtain
\begin{equation*}
\underset{n\rightarrow \infty }{\lim }\frac{M_{n}}{M_{n+1}}=\frac{%
x_{n+1}^{\ast }}{x_{n}^{\ast }}=\varphi .
\end{equation*}
\end{proof}

\begin{case}
\label{Case 2}Uniform distribution. Let $X_{0}$ and $X_{1}$ be iid with $%
Uniform(0,1)$ distribution. Then from (\ref{a00}) we obtain
\begin{eqnarray}
f_{X_{n}}(x) &=&\frac{1}{a_{n}a_{n-1}}\int\limits_{0}^{a_{n}}f_{X_{0}}(%
\frac{x-t}{a_{n-1}})dt=\frac{a_{n-1}}{a_{n}a_{n-1}}\int%
\limits_{0}^{a_{n}}dF_{X_{0}}(\frac{x-t}{a_{n-1}})  \notag \\
&=&\frac{1}{a_{n}}\left\{ F_{X_{0}}(\frac{x}{a_{n-1}})-F_{X_{0}}(\frac{%
x-a_{n}}{a_{n-1}})\right\}  \notag \\
&=&\left\{
\begin{array}{cc}
0, & x<0\text{ and }x>a_{n}+a_{n-1} \\
\frac{x}{a_{n}a_{n-1}} & 0\leq x\leq a_{n-1} \\
\frac{1}{a_{n}} & a_{n-1}\leq x\leq a_{n} \\
\frac{1}{a_{n}}(1-\frac{x-a_{n}}{a_{n-1}}) & a_{n}\leq x\leq a_{n}+a_{n-1}%
\end{array}%
\right. .  \label{a002}
\end{eqnarray}%
\includegraphics[scale=0.25]{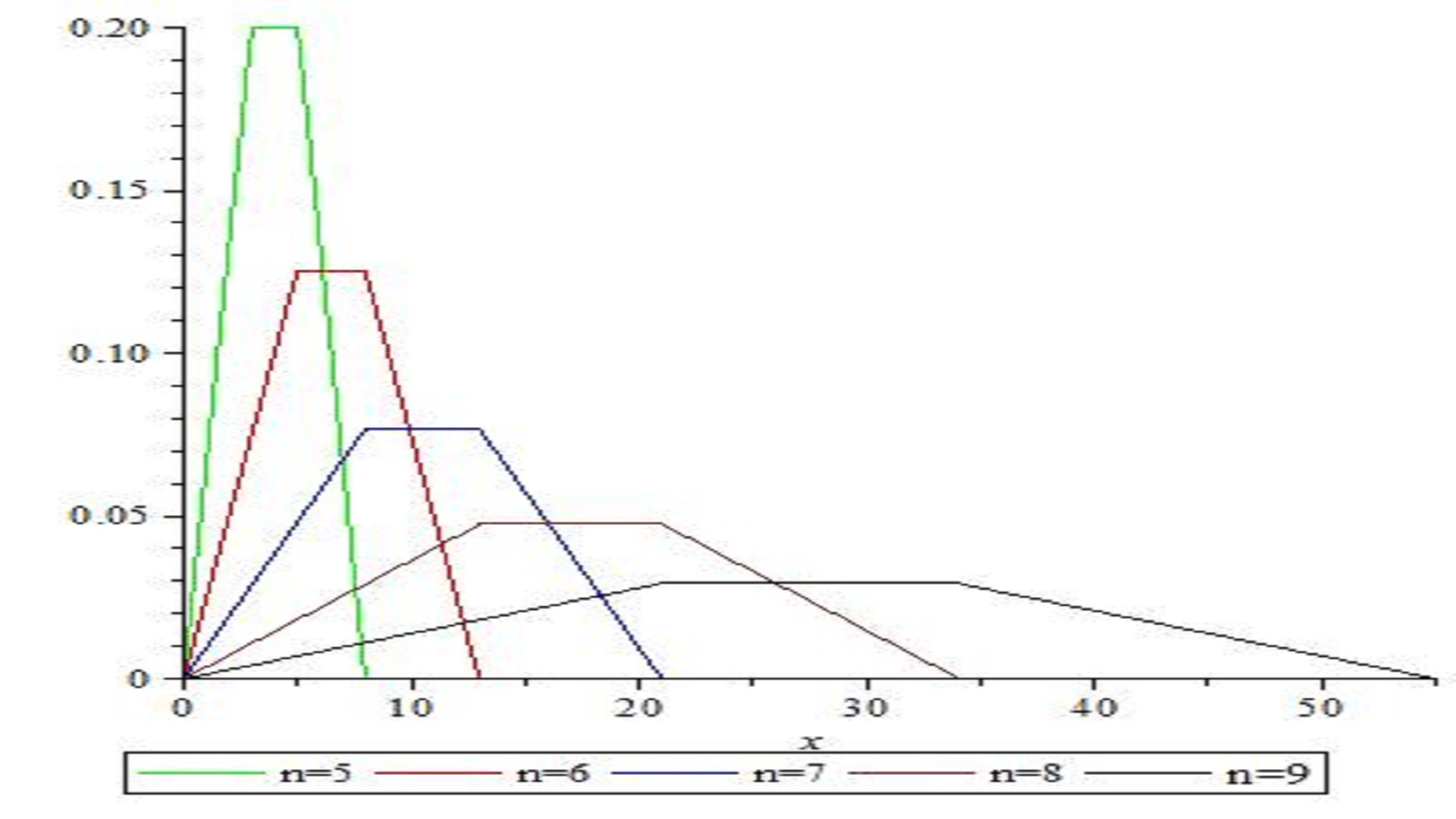}

\begin{tabular}{l}
Figure 2. Graphs of $f_{X_{n}}(x),n=5,6,7,8,9,$ given in (\ref{a002})
\end{tabular}
\\\\\\
It can be easily verified that $E(X_{n})=\frac{a_{n-1}+a_{n}}{2}=\frac{%
a_{n+1}}{2}$ and $var(X_{n})=\frac{a_{n-1}^{2}+a_{n}^{2}}{12}.$
\end{case}

\bigskip One can observe that $\ $ $f_{X_{n}}(x)$ is not unimodal, $%
f_{X_{n}}(x)$ is constant in the interval $(a_{n-1},a_{n})$ and $\ M_{n}=%
\underset{0<x<1}{\max }f_{X_{n}}(x)=\frac{1}{a_{n-1}}$ , $\ \underset{0<x<1}{%
\inf \arg \min }f_{X_{n}}(x)=a_{n-1},$ $\underset{0<x<1}{\sup \arg \min }%
f_{X_{n}}(x)=a_{n}$.

It is not difficult to observe that the similar to Theorem 1 results hold
also in this case.

\bigskip

\section{Large $n$ and normalized Fibonacci sequence of random variables}

Let $\ X_{n}=a_{n-1}X_{0}+a_{n}X_{1},$ $n=0,1,2,...$ \ be FSRV, where $X_{0}$
and $X_{1}$ are absolutely continuous random variables with joint pdf $%
f_{X_{0},X_{1}}(x,y),$ $(x,y)\in
\mathbb{R}
^{2}.$ Consider the sequence of random variables $Z_{n}\equiv \frac{X_{n+1}}{%
X_{n}},n=1,2,...$ . One has

\begin{equation*}
Z_{n}(\omega )=\frac{X_{n+1}(\omega )}{X_{n}(\omega )}=\frac{%
a_{n+1}X_{1}(\omega )+a_{n}X_{0}(\omega )}{a_{n}X_{1}(\omega
)+a_{n-1}X_{0}(\omega )}=\frac{\frac{a_{n+1}}{a_{n}}X_{1}(\omega
)+X_{0}(\omega )}{X_{1}(\omega )+\frac{a_{n-1}}{a_{n}}X_{0}(\omega )}=\frac{%
\frac{a_{n+1}}{a_{n}}X_{1}(\omega )+X_{0}(\omega )}{X_{1}(\omega )+\frac{1}{%
\frac{a_{n}}{a_{n-1}}}X_{0}(\omega )}.
\end{equation*}%
Since \ $\underset{n\rightarrow \infty }{\lim }\frac{a_{n+1}}{a_{n}}=\varphi
,$ ($\varphi =\frac{1-\sqrt{5}}{2}=1,6180339887...$ is the golden ratio), it
follows that
\begin{eqnarray*}
Z_{n}(\omega ) &\rightarrow &\varphi ,\text{ pointwise in } \\
\Omega _{1} &=&\{\omega :\varphi X_{1}+X_{0}\neq 0\text{ and }\varphi
X_{1}+X_{0}\neq \infty \}\subset \Omega \text{ }
\end{eqnarray*}%
For the normalized FSRV, the following limit relationship is valid.

\begin{theorem}
\label{Theorem 3}Let $E(X_{i})=\mu _{i},Var(X_{i})=\sigma _{i}^{2},i=0,1$
and
\begin{equation*}
Y_{n}(\omega )\equiv Y_{n}=\frac{X_{n}-E(X_{n})}{\sqrt{Var(X_{n})}}=\frac{%
X_{0}+\frac{a_{n}}{a_{n-1}}X_{1}-(\mu _{0}+\frac{a_{n}}{a_{n-1}}\mu _{1})}{%
\sqrt{\sigma _{0}^{2}+\frac{a_{n}^{2}}{a_{n-1}^{2}}\sigma _{1}^{2}}},\text{ }%
\omega \in \Omega .
\end{equation*}%
Then,
\begin{equation*}
Y_{n}\overset{}{\rightarrow }Y\equiv \frac{X_{0}+\varphi X_{1}-(\mu
_{0}+\varphi \mu _{1})}{\sqrt{\sigma _{0}^{2}+\varphi ^{2}\sigma _{1}^{2}}},%
\text{ as }n\rightarrow \infty \text{ for all }\omega \in \Omega .
\end{equation*}%
The limiting random variable $Y\equiv Y(\omega )$ \ has distribution
function (cdf)
\begin{equation}
P\{Y\leq x\}=P\{X_{0}+\varphi X_{1}\leq x\sqrt{\sigma _{0}^{2}+\varphi
^{2}\sigma _{1}^{2}}+(\mu _{0}+\varphi \mu _{1})\}.  \label{cc1}
\end{equation}
\end{theorem}

\bigskip

\bigskip

It is clear that the pdf of $X_{0}+\varphi X_{1}$ is
\begin{equation}
f_{X_{0}+\varphi X_{1}}(x)=\frac{1}{\varphi }\int\limits_{0}^{\infty
}f_{X_{0}}(x-t)f_{X_{1}}(\frac{t}{\varphi })dt.  \label{aa1}
\end{equation}%
and the pdf of $Y$ is then
\begin{equation}
f_{Y}(x)=\sqrt{\sigma _{0}^{2}+\varphi ^{2}\sigma _{1}^{2}}f_{X_{0}+\varphi
X_{1}}(x\sqrt{\sigma _{0}^{2}+\varphi ^{2}\sigma _{1}^{2}}+(\mu _{0}+\varphi
\mu _{1}))  \label{aaa1}
\end{equation}

\begin{example}
\label{Example 1}Let $\ X_{0}$ and $X_{1}$ be iid random variables having
exponential distribution with parameter $\lambda =1,$ then from (\ref{aa1})
we have
\begin{eqnarray*}
f_{X_{0}+\varphi X_{1}}(x) &=&\frac{1}{\varphi }\int\limits_{0}^{x}\exp
(-x-t)\exp (t-\frac{t}{\varphi })dt \\
&=&\frac{\exp (-x)}{\varphi -1}[\exp (x(1-\frac{1}{\varphi }))-1].
\end{eqnarray*}%
Therefore,
\begin{eqnarray*}
P\{Y &\leq &x\}=P\{X_{0}+\varphi X_{1}\leq x\sqrt{\sigma _{0}^{2}+\varphi
^{2}\sigma _{1}^{2}}+(\mu _{0}+\varphi \mu _{1})\} \\
&=&\int\limits_{0}^{c(x)}\left\{ \frac{\exp (-t)}{\varphi -1}(\exp (t(1-%
\frac{1}{\varphi }))-1)\right\} dt,
\end{eqnarray*}%
where $c(x)=x\sqrt{\sigma _{0}^{2}+\varphi ^{2}\sigma _{1}^{2}}+(\mu
_{0}+\varphi \mu _{1}).$ \ And the pdf is
\begin{eqnarray*}
f_{Y}(x) &=&\left\{
\begin{tabular}{ll}
$\sqrt{\sigma _{0}^{2}+\varphi ^{2}\sigma _{1}^{2}}\left\{ \frac{\exp (-c(x))%
}{\varphi -1}\left[ \exp \left( c(x)(1-\frac{1}{\varphi })\right) -1\right]
\right\} ,$ & $x\geq -\frac{(\mu _{0}+\varphi \mu _{1})}{\sqrt{\sigma
_{0}^{2}+\varphi ^{2}\sigma _{1}^{2}}}$ \\
$0$ & $Otherwise$%
\end{tabular}%
\right. \\
&=&\left\{
\begin{tabular}{ll}
$\sqrt{1+\varphi ^{2}}\left\{ \frac{\exp (-c(x))}{\varphi -1}\left[ \exp
\left( c(x)(1-\frac{1}{\varphi })\right) -1\right] \right\} ,$ & $x\geq -%
\frac{1+\varphi }{\sqrt{1+\varphi ^{2}}}$ \\
$0$ & $Otherwise$%
\end{tabular}%
\right. .
\end{eqnarray*}%
\includegraphics[scale=0.30]{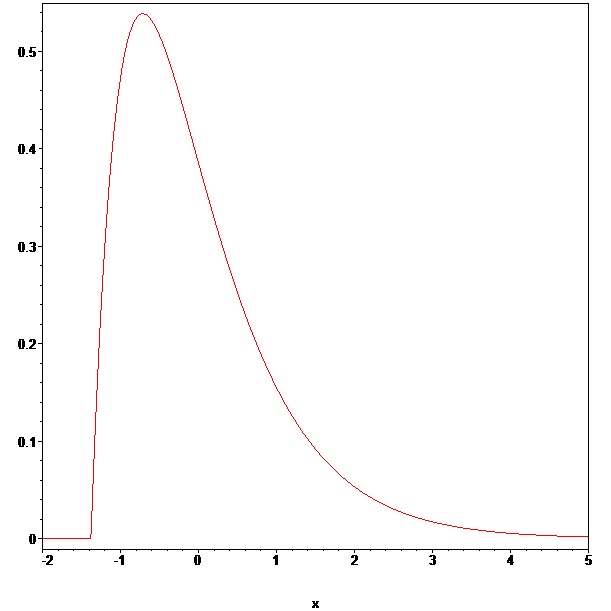}

\begin{tabular}{l}Figure 3. The graph of pdf $f_{Y}(x)$
\end{tabular}%
\\\\\

\end{example}

\begin{example}
\label{Example 2}\bigskip Let $X_{0}$ and $X_{1}$ be independent random
variables with $Uniform(0,1)$ distribution. Then from (\ref{aa1}) we have
\begin{equation*}
f_{X_{0}+\varphi X_{1}}(x)=\left\{
\begin{array}{ccc}
\frac{x}{\varphi }, & 0\leq x\leq 1 &  \\
\frac{1}{\varphi }, & 1\leq x\leq \varphi &  \\
\frac{1-x}{\varphi }+1, & \varphi \leq x\leq 1+\varphi &  \\
0, & elsehwere &
\end{array}%
\right. .
\end{equation*}%
This is a trapezoidal pdf with graph given below in Figure 4.
\\\\

\includegraphics[scale=0.20]{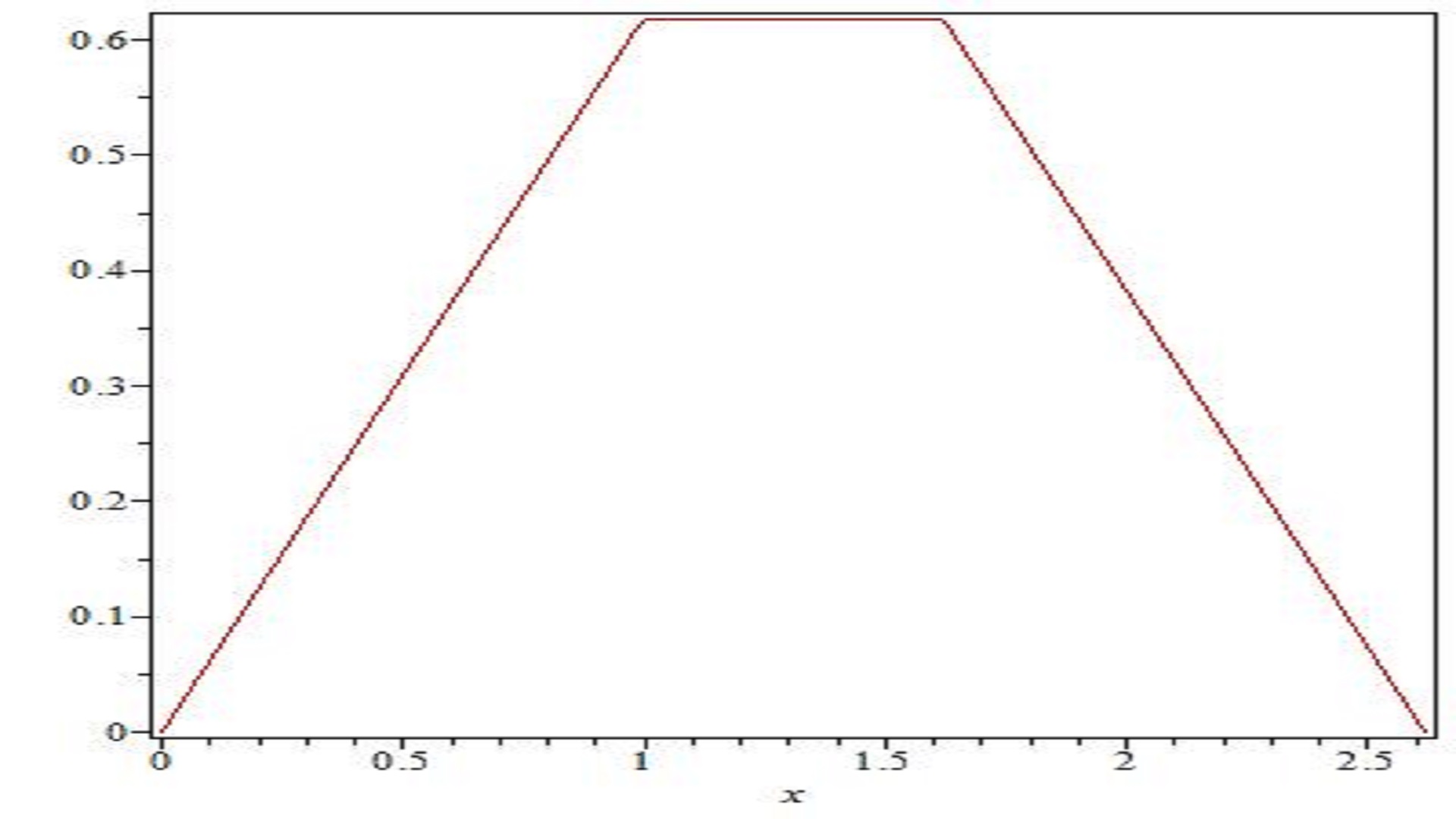}

\begin{tabular}{l}Figure 4. The graph of \ $f_{X_{0}+\varphi X_{1}}(x)$
\end{tabular}
\\\\\\


To find the distribution of limiting random variable $Y,$ we consider
\begin{equation*}
P\{Y\leq x\}=P\{X_{0}+\varphi X_{1}\leq x\sqrt{\sigma _{0}^{2}+\varphi
^{2}\sigma _{1}^{2}}+(\mu _{0}+\varphi \mu _{1})\}
\end{equation*}%
It is clear that
\begin{eqnarray*}
\mu _{0} &=&\mu _{1}=1/2,\text{ }\sigma _{0}^{2}=\sigma _{1}^{2}=1/12, \\
a &=&\sqrt{\sigma _{0}^{2}+\varphi ^{2}\sigma _{1}^{2}}=\sqrt{\frac{%
1+\varphi ^{2}}{12}},b=\mu _{0}+\varphi \mu _{1}=\frac{1+\varphi }{2}
\end{eqnarray*}
and the cdf of $Y$ is
\begin{eqnarray*}
F_{Y}(x) &=&P\{Y\leq x\}=P\{X_{0}+\varphi X_{1}\leq ax+b\}= \\
&&\left\{
\begin{tabular}{lll}
$0$ & $x\leq -\frac{b}{a}$ &  \\
$\frac{1}{\varphi }\int\limits_{0}^{ax+b}udu=\frac{(ax+b)^{2}}{2\varphi },$
& $-\frac{b}{a}\leq x\leq \frac{1-b}{a}$ &  \\
$\frac{1}{2\varphi }+\frac{1}{\varphi }\int\limits_{1}^{ax+b}du=\frac{1}{%
2\varphi }+\frac{ax+b-1}{\varphi },$ & $\frac{1-b}{a}\leq x\leq \frac{%
\varphi -b}{a}$ &  \\
$%
\begin{array}{c}
\frac{1}{2\varphi }+\frac{1}{\varphi }+\frac{1}{\varphi }\int\limits_{%
\varphi }^{ax+b}(\frac{1-u}{\varphi }+1)du \\
=\frac{2ax+2b+2ax\varphi +2b\varphi -a^{2}x^{2}-2axb-b^{2}-\varphi ^{2}-1}{%
2\varphi }%
\end{array}%
$ & $\frac{\varphi -b}{a}\leq x\leq \frac{1+\varphi -b}{a}$ &  \\
$1$ & $x\geq \frac{1+\varphi -b}{a}.$ &
\end{tabular}%
\right.
\end{eqnarray*}%

The pdf of $Y$ is
\begin{equation*}
f_{Y}(x)=\left\{
\begin{array}{cccc}
0, & x<-\frac{b}{a}\text{ \ or }x>\frac{1+\varphi -b}{a} &  &  \\
\frac{(ax+b)a}{\varphi }, & -\frac{b}{a}<x\leq \frac{1-b}{a} &  &  \\
\frac{a}{\varphi } & \frac{1-b}{a}<x\leq \frac{\varphi -b}{a} &  &  \\
\frac{a(1+\varphi -b-ax)}{\varphi } & \frac{\varphi -b}{a}<x\leq \frac{%
1+\varphi -b}{a} &  &
\end{array}%
\right. .
\end{equation*}%
%
\\\\
\
\end{example}

\bigskip\ \ \

\bigskip

\subsection{Limits of normalized sums of Fibonacci sequence of random
variables}

Here we are interested in the limiting behavior of sums of members of FSRV.
Consider $S_{n}=\sum\limits_{i=0}^{n}X_{i}$. \ We have
\begin{eqnarray*}
S_{n} &=&X_{0}+X_{1}+\cdots +X_{n}=X_{0}+X_{1}+\sum\limits_{i=2}^{n}X_{i} \\
&=&X_{0}+X_{1}+\sum\limits_{i=2}^{n}(a_{i-1}X_{0}+a_{i}X_{1}) \\
&=&X_{0}+X_{1}+X_{0}\sum\limits_{i=2}^{n}a_{i-1}+X_{1}\sum%
\limits_{i=2}^{n}a_{i} \\
&=&X_{0}+X_{1}+X_{0}\sum\limits_{i=1}^{n-1}a_{i}+X_{1}(\sum%
\limits_{i=1}^{n}a_{i}-a_{1}) \\
&=&X_{0}+X_{1}+X_{0}(a_{n+1}-1)+X_{1}(a_{n+2}-1-a_{1}) \\
&=&a_{n+1}X_{0}+(a_{n+2}-1)X.
\end{eqnarray*}%
Since%
\begin{equation*}
\sum\limits_{i=1}^{n}a_{i}=a_{n+2}-1.
\end{equation*}%
Therefore%
\begin{eqnarray*}
S_{n} &=&X_{0}+X_{1}+\cdots +X_{n} \\
&=&a_{n+1}X_{0}+(a_{n+2}-1)X_{1}.
\end{eqnarray*}%
The pdf of $S_{n}$ is
\begin{equation}
f_{S_{n}}(x)=\frac{1}{a_{n+1}(a_{n+1}-1)}\int\limits_{-\infty }^{\infty
}f_{X_{0}}(\frac{x-t}{a_{n+1}})f_{X_{1}}(\frac{t}{a_{n+2}-1)})dt.  \label{q3}
\end{equation}

\begin{theorem}
\bigskip \label{Theorem 4}Under conditions of Theorem 3 for a sequence $%
X_{0},X_{1},$ $X_{n}=a_{n-1}X_{0}+a_{n}X_{1},$ $n=2,3,...$ $\ \ $we have%
\begin{eqnarray*}
ES_{n} &=&a_{n+1}\mu _{0}+(a_{n+2}-1)\mu _{1} \\
Var(S_{n}) &=&a_{n+1}^{2}\sigma _{0}^{2}+(a_{n+2}-1)^{2}\sigma _{1}^{2}
\end{eqnarray*}%
\begin{equation*}
\frac{S_{n}-E(S_{n})}{\sqrt{var(S_{n})}}\rightarrow Y\text{ as }n\rightarrow
\infty ,\text{ for all }\omega \in \Omega ,
\end{equation*}%
where $Y$ has cdf \ (\ref{cc1}).
\end{theorem}

\begin{proof}
Indeed,
\begin{eqnarray*}
&&\frac{S_{n}-E(S_{n})}{\sqrt{var(S_{n})}} \\
&=&\frac{X_{0}+(\frac{a_{n+2}}{a_{n+1}}-\frac{1}{a_{n+1}})X_{1}-(\mu _{0}+(%
\frac{a_{n+2}}{a_{n+1}}-\frac{1}{a_{n+1}})\mu _{1}}{\sqrt{\sigma _{0}^{2}+(%
\frac{a_{n+2}}{a_{n+1}}-\frac{1}{a_{n+1}})^{2}\sigma _{1}^{2}},} \\
&&\overset{}{\rightarrow }\frac{X_{0}+\varphi X_{1}-(\mu _{0}+\varphi \mu
_{1})}{\sqrt{\sigma _{0}^{2}+\varphi ^{2}\sigma _{1}^{2}}}=Y,\text{ as }%
n\rightarrow \infty .
\end{eqnarray*}
\end{proof}

\bigskip

\begin{example}
\label{Example 3}Let \ $X_{0}$ and $X_{1}$ be iid exponential(1) random
variables. Then the pdf of $S_{n}$ is
\begin{eqnarray}
f_{S_{n}}(x) &=&\frac{1}{a_{n+1}(a_{n+2}-1)}\int\limits_{0}^{x}\exp (\frac{%
x-t}{a_{n+1}})\exp (\frac{t}{a_{n+2}-1)})dt  \notag \\
&=&\frac{\exp (-\frac{x}{a_{n+1}})}{a_{n+1}-a_{n+2}+1}\left( 1-\exp
(-x\left( \frac{1}{a_{n+2}-1}-\frac{1}{a_{n+1}}\right) \right) .  \label{q3a}
\end{eqnarray}%
\end{example}

\bigskip

\bigskip

\section{Joint distributions of $X_{n}$ and $X_{n+k}$}

Next, we focus on the joint distributions of $X_{n}=a_{n-1}X_{0}+a_{n}X_{1}$
and $X_{n+k}=a_{n+k-1}X_{0}+a_{n+k}X_{1},$ for $k\geq 1.$

\begin{theorem}
\label{Theorem 5}The joint pdf of $X_{n}$ and $X_{n+k}$ is%
\begin{eqnarray}
&&f_{X_{n},X_{n+k}}(y_{0},y_{1})  \notag \\
&=&\frac{1}{a_{k}}f_{X_{0},X_{1}}(\frac{a_{n+k}y_{0}-y_{1}a_{n}}{%
(-1)^{n}a_{k}},\frac{a_{n-1}y_{1}-a_{n+k-1}y_{0}}{(-1)^{n}a_{k}}).
\label{bb2}
\end{eqnarray}
\end{theorem}

\begin{proof}
Let%
\begin{equation}
\left\{
\begin{array}{c}
y_{0}=a_{n-1}x_{0}+a_{n}x_{1} \\
y_{1}=a_{n+k-1}x_{0}+a_{n+k}x_{1}%
\end{array}%
\right. .  \label{b1}
\end{equation}%
The Jacobian of this linear transformation is \ $%
J=a_{n-1}a_{n+k}-a_{n}a_{n+k-1}$ and the solution of the system of equations
(\ref{b1}) is
\begin{equation*}
\left\{
\begin{array}{c}
x_{0}=(a_{n+k}y_{0}-y_{1}a_{n})/(a_{n-1}a_{n+k}-a_{n}a_{n+k-1}) \\
x_{1}=(a_{n-1}y_{1}-a_{n+k-1}y_{0})/(a_{n-1}a_{n+k}-a_{n}a_{n+k-1})%
\end{array}%
\right. .
\end{equation*}%
Therefore, the joint pdf of $X_{n}$ and $X_{n+k}$ is
\begin{eqnarray}
&&f_{X_{n},X_{n+k}}(y_{0},y_{1})  \notag \\
&=&\frac{1}{\left\vert a_{n-1}a_{n+k}-a_{n}a_{n+k-1}\right\vert }%
f_{X_{0},X_{1}}(\frac{a_{n+k}y_{0}-y_{1}a_{n}}{a_{n-1}a_{n+k}-a_{n}a_{n+k-1}}%
,  \notag \\
&&\frac{a_{n-1}y_{1}-a_{n+k-1}y_{0}}{a_{n-1}a_{n+k}-a_{n}a_{n+k-1}}).
\label{b2}
\end{eqnarray}%
Using the d'Ocagne's identity \ \ (see e.g. Dickson (1966)) $\
a_{m}a_{n+1}-a_{m+1}a_{n}=(-1)^{n}a_{m-n}$ we have \ $%
J=a_{n-1}a_{n+k}-a_{n}a_{n+k-1}=-(a_{n+k-1}a_{n}-a_{n+k}a_{n-1})=(-1)^{n}a_{k}.
$ Therefore,%
\begin{eqnarray*}
&&f_{X_{n},X_{n+k}}(y_{0},y_{1}) \\
&=&\frac{1}{a_{k}}f_{X_{0},X_{1}}(\frac{a_{n+k}y_{0}-y_{1}a_{n}}{%
(-1)^{n}a_{k}},\frac{a_{n-1}y_{1}-a_{n+k-1}y_{0}}{(-1)^{n}a_{k}}).
\end{eqnarray*}
\end{proof}

\bigskip

\begin{corollary}
If $X_{0}$ and $X_{1}$ are independent then
\begin{eqnarray}
&&f_{X_{n},X_{n+k}}(x,y)  \notag \\
&=&\frac{1}{a_{k}}f_{X_{0}}\left( \frac{a_{n+k}x-ya_{n}}{(-1)^{n}a_{k}}%
\right) f_{X_{1}}\left( \frac{a_{n-1}y-a_{n+k-1}x}{(-1)^{n}a_{k}}\right) .
\label{bb3}
\end{eqnarray}
\end{corollary}

\begin{example}
\label{Example 4}Let $X_{0}$ and $X_{1}$ be iid exponential(1) random
variables, $n=4,k=3.$ Then \ $a_{n+k}=a_{7}=13,$ $a_{n+k-1}=a_{6}=8,$ $%
a_{n-1}=a_{3}=2,$ $a_{n}=a_{4}=3$ $\ $and $a_{k}=a_{3}=2.$ Then from (\ref%
{bb3})
\begin{eqnarray}
&&f_{X_{4},X_{7}}(x,y)  \notag \\
&=&\left\{
\begin{array}{cc}
\begin{array}{c}
\frac{1}{2}\exp (-(13/2)x+(3/2)y) \\
\times \exp (-y+4x),%
\end{array}
& x\geq 0\text{ and }4x\leq y\leq 13/3x \\
0 & otherwise%
\end{array}%
\right. .  \notag \\
&=&\left\{
\begin{array}{cc}
\frac{1}{2}\exp (-(5/2)x)\exp (y/2) & x\geq 0\text{ and }4x\leq y\leq 13/3x
\\
0 & otherwise%
\end{array}%
\right.  \label{c4}
\end{eqnarray}%
The marginal pdf's are
\begin{equation}
f_{X_{4}}(x)=\left\{
\begin{array}{cc}
e^{-\frac{x}{3}}-e^{-\frac{x}{2}}, & x\geq 0 \\
0, & otherwise%
\end{array}%
\right.  \label{m1}
\end{equation}%
and%
\begin{equation*}
f_{X_{7}}(x)=\left\{
\begin{array}{cc}
\frac{1}{5}\left( e^{-\frac{x}{13}}-e^{-\frac{x}{8}}\right) , & x\geq 0 \\
0, & otherwise%
\end{array}%
\right. .
\end{equation*}
\end{example}

\begin{example}
\label{Example 5}Let $X_{0}$ and $X_{1}$ be independent uniform(0,1) random
variables. Again, let $n=4,k=3.$ Then \ $a_{n+k}=a_{7}=13,$ $%
a_{n+k-1}=a_{6}=8,$ $a_{n-1}=a_{3}=2,$ $a_{n}=a_{4}=3$ $\ $and $%
a_{k}=a_{3}=2.$ Then
\begin{eqnarray}
&&f_{X_{n},X_{n+k}}(x,y)  \notag \\
&=&\frac{1}{a_{k}}f_{X_{0}}\left( \frac{a_{n+k}x-ya_{n}}{(-1)^{n}a_{k}}%
\right) f_{X_{1}}\left( \frac{a_{n-1}y-a_{n+k-1}x}{(-1)^{n}a_{k}}\right)
\notag \\
&=&\left\{
\begin{array}{cc}
\frac{1}{a_{k}} & 0\leq \frac{a_{n+k}x-ya_{n}}{(-1)^{n}a_{k}}\leq 1,0\leq
\frac{a_{n-1}y-a_{n+k-1}x}{(-1)^{n}a_{k}}\leq 1 \\
0, & otherwise%
\end{array}%
\right.  \label{c1}
\end{eqnarray}%
(To check whether (\ref{c1}) is a pdf, \ we need to show $%
\int\limits_{0}^{1}\int\limits_{0}^{1}f_{X_{n},X_{n+k}}(x,y)dxdy=1.$
Indeed, \ \
\begin{eqnarray*}
&&\int\limits_{0}^{1}\int\limits_{0}^{1}f_{X_{n},X_{n+k}}(x,y)dxdy \\
&=&\frac{1}{a_{k}}\underset{0\leq \frac{a_{n+k}x-ya_{n}}{(-1)^{n}a_{k}}\leq
1,0\leq \frac{a_{n-1}y-a_{n+k-1}x}{(-1)^{n}a_{k}}\leq 1}{\int \int }dxdy \\
&=&\left\{
\begin{array}{c}
\begin{array}{c}
a_{n+k}x-ya_{n}=t,a_{n-1}y-a_{n+k-1}x=s \\
x=\frac{ta_{n-1}+sa_{n}}{(-1)^{n}a_{k}},y=\frac{sa_{n+k}+ta_{n+k-1}}{%
(-1)^{n}a_{k}}%
\end{array}
\\
t\leq (-1)^{n}a_{k},s\leq (-1)^{n}a_{k} \\
J=\left\vert
\begin{array}{cc}
\frac{a_{n-1}}{(-1)^{n}a_{k}} & \frac{a_{n}}{(-1)^{n}a_{k}} \\
\frac{a_{n+k-1}}{(-1)^{n}a_{k}} & \frac{a_{n+k}}{(-1)^{n}a_{k}}%
\end{array}%
\right\vert =\frac{a_{n-1}a_{n+k}-a_{n}a_{n+k-1}}{(-1)^{2n}a_{k}^{2}}=\frac{%
(-1)^{n}a_{k}}{(-1)^{2n}a_{k}^{2}}%
\end{array}%
\right\} \\
&=&\frac{1}{a_{k}}\int\limits_{0}^{(-1)^{n}a_{k}}\int%
\limits_{0}^{(-1)^{n}a_{k}}\frac{1}{\left\vert (-1)^{n}a_{k}\right\vert }%
dxdy=1.)
\end{eqnarray*}%
\ For $n=4$ and $k=3,$ the
\begin{eqnarray}
&&f_{X_{4},X_{7}}(x,y)  \notag \\
&=&\frac{1}{2}f_{X_{0}}\left( \frac{13x-3y}{2}\right) f_{X_{1}}\left( \frac{%
2y-8x}{2}\right)  \notag \\
&=&\left\{
\begin{array}{cc}
\frac{1}{2}, & 0\leq \frac{13x-3y}{2}\leq 1,0\leq \frac{2y-8x}{2}\leq 1 \\
0, & otherwise%
\end{array}%
\right. .  \label{c2}
\end{eqnarray}%
\end{example}

\section{Prediction of future values}

\bigskip It is well known that with respect to squared error loss, the best
unbiased predictor of $X_{n+k},$ given $X_{n}$ \ is%
\begin{equation*}
E\{X_{n+k}\mid X_{n}\}.
\end{equation*}%
Let
\begin{eqnarray}
g(x) &=&E\{X_{n+k}\mid X_{n}=x\}  \notag \\
&=&\frac{1}{f_{X_{n}}(x)}\int\limits_{-\infty }^{\infty
}yf_{X_{n},X_{n+k}}(x,y)dy,  \label{cc2}
\end{eqnarray}%
then $E\{X_{n+k}\mid X_{n}\}=g(X_{n}).$ Using (\ref{a000}) and (\ref{bb2}) $%
\ $\ from (\ref{cc2}) one can easily calculate the best predictor of $%
X_{n+k},$ given $X_{n}.$

\begin{example}
\label{Example 6}Let $X_{0}$ and $X_{1}$ be independent exponential(1)
random variables. Let $n=4,k=3.$ Then \ $a_{n+k}=a_{7}=13,$ $%
a_{n+k-1}=a_{6}=8,$ $a_{n-1}=a_{3}=2,$ $a_{n}=a_{4}=3$ $\ $and $%
a_{k}=a_{3}=2 $ as in Example 4. Then from (\ref{c4}) we can write
\begin{eqnarray*}
g(x) &=&\frac{1}{e^{-\frac{x}{3}}-e^{-\frac{x}{2}}}\int\limits_{4x}^{13/3x}y%
\frac{1}{2}\exp (-(5/2)x)\exp (y/2)dy \\
&=&\frac{1}{3}\frac{12e^{-x/2}-6e^{-x/2}+6e^{-x/3}-13e^{-x/3}}{%
e^{-x/2}-e^{-x/3}} \\
&=&4x-2-\frac{x}{3(e^{-x/6}-1)},
\end{eqnarray*}
\end{example}

Therefore,
\begin{equation*}
X_{7}\simeq 4X_{4}-2-\frac{X_{4}}{3(e^{-X_{4}/6}-1).}.
\end{equation*}

\begin{conclusion}
In this note, we considered the sequence of random variables $\{X_{0},X_{1},$
$X_{n}=X_{n-2}+X_{n-1},$ $n=2,3,..\}$ which is equivalent to  $\left\{
X_{0},X_{1},X_{n}=\
a_{n-1}X_{0}+a_{n}X_{1},n=2,3,...\right\} ,$ where $X_{0}$
and $X_{1}$ are \ absolutely continuous random variables with joint pdf $%
f_{X_{0},X_{1}},$ and $a_{n}=a_{n-1}+a_{n-2},$ $n=2,3,...$ $(a_{0}=0,$ $%
a_{1}=1)$ is the Fibonacci sequence. In the paper, the sequence $%
X_{n},n=0,1,2,...$ \ is referred to as the Fibonacci Sequence of Random \
Variables. We investigated the limiting properties of some ratios and
normalizing sums of this sequence. For exponential and uniform distribution
cases, we derived some interesting limiting properties that reduce to the
golden ratio and also investigated the joint distributions of $X_{n}$ and $%
X_{n+k}.$ The considered random sequence has benefical properties and may be
worthy of attention associated with random sequences and autoregressive
models. \ \
\end{conclusion}

\section{Appendix}

For illustration of the behaviour of FSRV, the simulated values of random
variables $X_{0}$ and $X_{1}$ from uniform (0,1) and standard normal
distribution are obtained. The corresponding codes in R are also given. The
corresponding code in R for uniform(0,1) distribution is:
\\\\
$>$a$<-$seq(1:10); for (i in 3:10) a[i]=a[i-1]+a[i-2]; x$<-$runif(10); y$<-$runif(10); z$<-$numeric(10); for (i in 2:10) z[i]=a[i]*x[i]+a[i-1]*x[i-1]; c$<-$seq(1:10); plot(c,z,col="red",bg="yellow",pch=22,bty="l");
\\\\
The corresponding code in R for standard normal distribution is:
\\\\
$>$a$<-$seq(1:10); for (i in 3:10) a[i]=a[i-1]+a[i-2]; x$<-$rnorm(10); y$<-$rnorm(10); z$<-$numeric(10); for (i in 2:10) z[i]=a[i]*x[i]+a[i-1]*x[i-1]; c$<-$seq(1:10); plot(c,z,col="red",bg="yellow",pch=22,bty="l");
\\\\


\bigskip


\bigskip

\begin{thebibliography}{9}
\bibitem{} Dickson. L. E. (1966)\textit{\ History of the Theory of Numbers}
Volume 1, New York: Chelsea.

\bibitem{} Gnedenko, B.V. (1978) \textit{The Theory of Probability}, Mir
Publishers, Moscow.

\bibitem{} Feller, W. (1971) \textit{An Introduction to Probability Theory
and Its Applications}, Volume 2, John Wiley \& Sons Inc. , New York, London,
Sydney.

\bibitem{} Melham, R.S. and Shannon, A.G. \ (1995) \textit{A generalization
of the Catalan identity and some consequences}, The Fibonacci Quarterly 33,
82--84, 1995.

\bibitem{} Ross, S. (2016) \textit{A First Course in Probability}.
Prentice-Hall Inc. , NJ.

\bibitem{} Skorokhod, A.V. (2005) \textit{Basic Principles and Applications
of Probability Theory}, Springer.
\end{thebibliography}
\end{document}